\newif \ifcomments \commentsfalse
\newif \iffull \fulltrue

\fulltrue
\documentclass[11pt]{article}
\usepackage[utf8]{inputenc}
\usepackage[letterpaper,margin=1in,bmargin=1.5in]{geometry}
\usepackage[dvipsnames]{xcolor}

\usepackage{authblk}
\usepackage{setspace}
\usepackage{tikz}
\usetikzlibrary{calc,math,arrows,arrows.meta,shapes,positioning,matrix,decorations.pathreplacing}
\usepackage{amsfonts,amsmath,amsthm}
\usepackage{enumitem}
\usepackage{cryptocode} 
\usepackage{multirow, ctable, tabu, makecell, graphics}
\usepackage{algorithm}
\usepackage[noend]{algpseudocode}
\usepackage[style=numeric-comp,backend=bibtex]{biblatex}
\usepackage[pdfauthor={},pdfpagelabels=true,linktocpage=true,hidelinks,linktoc=all,colorlinks]{hyperref}
\usepackage[nameinlink,capitalize]{cleveref}
\hypersetup{linkcolor = {black}, citecolor = {magenta}, urlcolor = {black}}
\usepackage{amsmath}
\usepackage{amsfonts,amssymb,amsthm}
\usepackage{pifont}
\usepackage{xspace}
\usepackage{bm}
\usepackage{enumitem}
\usepackage{hyperref}
\usepackage{biblatex}

\newcommand{\sysname}{Zef\xspace}

\makeatletter
\newcommand{\shorteq}{%
  \settowidth{\@tempdima}{e}
  \resizebox{\@tempdima}{\height}{=}%
}
\makeatother

\newcommand{\keyword}[1]{{\normalfont \textsf{#1}}}

\newcommand{\pk}{\keyword{pk}}

\newcommand{\id}{\keyword{id}}
\newcommand{\swid}{\keyword{swid}}

\newcommand{\cert}{\keyword{cert}}
\newcommand{\auth}{\keyword{auth}}
\newcommand{\decision}{\keyword{decision}}
\newcommand{\round}{\keyword{round}}
\newcommand{\proposed}{\keyword{proposed}}
\newcommand{\locked}{\keyword{locked}}

\newcommand{\ChangeKey}{\keyword{ChangeKey}}
\newcommand{\OpenAccount}{\keyword{OpenAccount}}

\newcommand{\Transfer}{\keyword{Transfer}}
\newcommand{\Abort}{\keyword{Abort}}
\newcommand{\Confirm}{\keyword{Confirm}}
\newcommand{\HandleProposal}{\keyword{HandleProposal}}
\newcommand{\HandlePreCommit}{\keyword{HandlePreCommit}}
\newcommand{\HandleCommit}{\keyword{HandleCommit}}
\newcommand{\Proposal}{\keyword{Proposal}}
\newcommand{\Commit}{\keyword{Commit}}
\newcommand{\PreCommit}{\keyword{PreCommit}}
\newcommand{\Lock}{\keyword{Lock}}
\newcommand{\Execute}{\keyword{Execute}}
\newcommand{\StartConsensusInstance}{\keyword{StartConsensusInstance}}
\newcommand{\LockInto}{\keyword{LockInto}}

\newcommand{\initstate}{\keyword{init}}
\newcommand{\accounts}{\keyword{accounts}}
\newcommand{\atomicswaps}{\keyword{atomic\_swaps}}
\newcommand{\val}{\keyword{value}}

\newcommand{\nextsequence}{\keyword{next\_sequence}}
\newcommand{\confirmed}{\keyword{confirmed}}
\newcommand{\received}{\keyword{received}}
\newcommand{\pending}{\keyword{pending}}

\newcommand{\balance}{\keyword{balance}}

\newcommand{\Setup}{\keyword{Setup}}
\newcommand{\Encrypt}{\keyword{Encrypt}}
\newcommand{\ShareDecrypt}{\keyword{ShareDecrypt}}
\newcommand{\ShareVerify}{\keyword{ShareVerify}}
\newcommand{\Combine}{\keyword{Combine}}

\newcommand{\PK}{\keyword{PK}}
\newcommand{\VK}{\keyword{VK}}
\newcommand{\SK}{\keyword{SK}}

\ifcomments
    \newcommand{\alberto}[1]{\textsf{\color{violet}{[Alberto: {#1}]}}}
    \newcommand{\mathieu}[1]{\textsf{\color{olive}{[Mathieu: {#1}]}}}
    \newcommand{\george}[1]{\textsf{\color{orange}{[George: {#1}]}}}
    \newcommand{\michal}[1]{\textsf{\color{teal}{[Michał: {#1}]}}}
\else
    \newcommand{\alberto}[1]{}
    \newcommand{\mathieu}[1]{}
    \newcommand{\george}[1]{}
    \newcommand{\michal}[1]{}

\fi

\newtheorem{theorem}{Theorem}[section]

\newtheorem{lemma}[theorem]{Lemma}

\theoremstyle{definition}

\theoremstyle{remark}

\newcommand*\circled[1]{\tikz[baseline=(char.base)]{
            \node[scale=0.8,shape=circle,fill=black!75,inner sep=2pt] (char) {\textcolor{white}{\textbf{#1}}};}}

\newcommand*\bluecircled[1]{\tikz[baseline=(char.base)]{
            \node[scale=0.8,shape=circle,fill=BlueViolet!75,inner sep=2pt] (char) {\textcolor{white}{\textbf{#1}}};}}

\begin{document}
\title{Low-latency, Scalable, DeFi with Zef}
\author{
Mathieu Baudet$^1$, Alberto Sonnino$^1$, Micha\l{} Kr\'{o}l$^2$
}
\date{
    $^1$Novi, Facebook\\%
    $^2$City, Univeristy of London\\[2ex]%
}

\maketitle



\section{Introduction}

\sysname~\cite{zef-report} was recently proposed to extend the low-latency, Byzantine-Fault Tolerant (BFT) payment protocol FastPay~\cite{fastpay} with anonymous coins.
This report explores further extensions of FastPay and \sysname beyond payments.
We start by off-chain assets (e.g. NFTs) in Section~\ref{sec:assets}. We introduce the idea of on-demand BFT consensus instances throught the example of atomic swaps between account owners in Section~\ref{sec:atomic_swap}.



\section{Off-chain Assets}
\label{sec:assets}

In Zef, a coin is defined as a quorum of validator signatures (known as a \emph{certificate}) that binds a user account identifier $\id$ with a monetary value $v$. For instance, the \emph{transparent coins} of Zef are defined as $C = \cert[(\id, s, v)]$ for some seed $s$. To provide unlinkability and privacy, Zef also defines another type of coin called \emph{opaque coins} along the same lines but using the Coconut scheme~\cite{coconut}. For privacy reasons, all Zef coins (transparent and opaque) are stored off-chain\footnote{We use the expressions ``off-chain" and ``on-chain" for the data outside and inside the Zef authorities, although Zef is not blockchain, strictly speaking.} by their owners.

The coins linked to an account $\id$ can be spent altogether by the owner of $\id$ by issuing an operation $\mathsf{Spend}$ that deactivates $\id$ permanently.
Compared to FastPay, Zef accounts are addressed by a unique identifier $\id$ that can never be replayed in future accounts. This makes it possible to effectively remove deactivated accounts and avoid a permanent storage cost every time that coins are spent.
Zef coins can be generalized into \emph{assets}, where a certificate binds arbitrary data to an account~$\id$:

\begin{itemize}
\item Assuming that a distinct data value $x$ needs to be linked to an account $\id$, the asset may simply defined as $A = \cert[(\id, x)]$.

\item The procedure to consume and create coins in Zef can be generalized to consume input assets and create new assets according to specific rules of the form
\[ (x^{out}_1, \ldots, x^{out}_d) = f_\mathsf{exec}(P, x^{in}_1, \ldots, x^{in}_\ell) \] where $f_\mathsf{exec}$ is a fixed, deterministic, partial \emph{execution function} and $P$ is a set of parameters.

\item Importantly, the $\mathsf{Spend}$ operations used to deactivate the inputs accounts $x^{in}_1, \ldots, x^{in}_\ell$ must contain a commitment on $P$ so that any replay of the asset creation request on the same input assets produces exactly the same output assets. (See the coin creation request in Zef~\cite{zef-report})

\end{itemize}

Off-chain assets based on Zef provide storage-free certified execution at scale for single-owner data. 
This new general framework applies in particular to Non-Fungible Tokens (NFTs) with the following benefits:
\begin{itemize}
    \item Off-chain storage provides some level of privacy w.r.t other users.
    \item Any owner-initiated operations such as transferring, combining NFTs, and applying legitimate modifications are supported at scale.
\end{itemize}


\section{Atomic Swaps}
\label{sec:atomic_swap}



We now describe an extension of Zef~\cite{zef-report} meant to support the swap of ownership of two accounts in an atomic way.

To prevent race-conditions between operations (say, spending and swapping assets), a correct solution must start by requiring each owner to independently lock their asset into a new instance of the swap protocol. This creates a difficulty as one owner may lock their asset while the other fails to do so, or simply changes their mind. Hence, contrary to the operations described in \sysname, an acceptable solution for atomic swap must support two possible outcomes: confirm and abort.

Because authorities must agree on this binary outcome, this raises the interesting question whether a correct solution for atomic swap in the FastPay model must implement a fully-featured, one-shot binary consensus protocol. From the FLP theorem~\cite{flp}, we know that a deterministic, asynchronous solution for a fault-tolerant consensus cannot guarantee both safety and eventual termination.

In this section we start by describing an implementation of atomic swap where termination assumes eventual cooperation between the owners of the two accounts. Importantly, this assumption is only made after both owners have locked their assets. We will discuss an optional refinement of the protocol at the end in order to enforce eventual termination in the partially-synchronous model.

\paragraph{Atomic-swap instances.} We augment the state of each authority $\alpha$ (Zef~\cite{zef-report}, Section 3 and 4) with a new field $\atomicswaps(\alpha)$ that maps certain UIDs to the states of ongoing atomic-swap instances. Although atomic-swap instances are addressed by a UID, they have a distinct type (see below) and are not subject to account operations of Section~4 of Zef~\cite{zef-report}. For clarity, below, we use $\swid$ (rather than $\id$) to denote identifiers in the domain of $\atomicswaps(\alpha)$.

\paragraph{New account operations.}
To create atomic-swap instances and lock assets into an ongoing instances, we extend the protocol of Section~4 of Zef~\cite{zef-report} with two new account operations:
\begin{itemize}
    \item To lock the content of an account $\id$ into an atomic-swap instance $\swid$, we introduce operations of the form $O = \LockInto(\swid, i, \pk)$ where $i \in \{ 1, 2 \}$ is the role index in the atomic swap, and $\pk$ is a key provided for authentication purposes and to be set on the other account in case of success. Such operation is sent in a (locking) request $R = \Lock(\id, n, O)$.
    
    \item To create a new instance of an atomic swap identified by a fresh identifier $\swid = \id :: n$, a (regular) request $R = \Execute(\id, n, O)$ may be used with \[O = \StartConsensusInstance(\swid, \id_1, n_1, \id_2, n_2)\] meaning that the owners of $\id_1$ and $\id_2$ ($\id_1 \neq \id_2$) wish to atomically exchange their ownership of $\id_1$ and $\id_2$. The numbers $n_1$ and $n_2$ are the expected sequence numbers of the lock certificates of the respective owners for the operation $\LockInto$ above.
\end{itemize}

This new operations are summarized in Algorithm~\ref{alg:account_operations}. We recall the framework for account requests in Zef in Algorithm~\ref{alg:account_service}.

\begin{algorithm}[t]
\caption{Account operations (core Zef + atomic swap)}\label{alg:account_operations}
\small
\begin{algorithmic}[1]
\algdef{SE}[ASYNC]{Async}{EndAsync}{\textbf{do asynchronously}}{\algorithmicend}%
\algtext*{EndAsync}%
\algnewcommand\algorithmicswitch{\textbf{switch}}
\algnewcommand\algorithmiccase{\textbf{case}}
\algnewcommand\algorithmicassert{\texttt{assert}}
\algnewcommand\Assert[1]{\State \algorithmicassert(#1)}%
\algdef{SE}[SWITCH]{Switch}{EndSwitch}[1]{\algorithmicswitch\ #1\ \algorithmicdo}{\algorithmicend\ \algorithmicswitch}%
\algdef{SE}[CASE]{Case}{EndCase}[1]{\algorithmiccase\ #1:}{\algorithmicend\ \algorithmiccase}%
\algtext*{EndSwitch}%
\algtext*{EndCase}%

\Function{ValidateOperation}{$\id$, $n$, $O$} \Comment{Internal validation of account operation}
\Switch{$O$}
    \Case{$\OpenAccount(\id', \pk')$}
        \State ensure $\id' = \id :: \nextsequence^\id$
    \EndCase
    \Case {$\Transfer(\id', \val)$}
        \State ensure $0 < \val \leq \balance^\id$
    \EndCase
    \Case {$\ChangeKey(\pk')$}
        \State pass
    \EndCase
    \Case {$\StartConsensusInstance(\swid, \id_1, \pk_1, \id_2, \pk_2)$}
        \State ensure $\swid = \id :: \nextsequence^\id$
        \State ensure $\id_1 \neq \id_2$
    \EndCase
    \Case {$\LockInto(\swid, i, pk)$} \Comment{Temporarily transfer the management of $\id$ to $\swid$}
        \State \textbf{return} $\Lock(\id, n, O)$ \Comment{$O$ is valid and locking.}
    \EndCase
\EndSwitch
\State \textbf{return} $\Execute(\id, n, O)$ \Comment{If we reach this, $O$ is valid and regular.}
\EndFunction

\Statex

\Function{ExecuteOperation}{$\id$, $O$, $C$}\Comment{Execution of account operation (unchanged from Zef)}
\Switch{$O$}
    \Case{$\OpenAccount(\id', \pk')$}
        \color{BlueViolet}
        \Async \Comment{Cross-shard request to $\id'$}
            \State run \Call{InitAccount}{\id', \pk'} \Comment{Create new account}
            \State $\received^{\id'} \gets \received^{\id'} :: C$ \Comment{Log certified request in recipient's account}
        \EndAsync
        \color{black}
    \EndCase
    \Case {$\Transfer(\id', \val)$}
        \State $\balance^\id \gets \balance^\id - \val$ \Comment{Update sender's balance}
        \color{BlueViolet}
        \Async \Comment{Cross-shard request to $\id'$}
            \If {$\id' \not\in \accounts$}
                \State run \Call{InitAccount}{$\id'$, $\bot$} \Comment{Create receiver's account if needed}
            \EndIf
            \State $\balance^{\id'} \gets \balance^{\id'} + \val$ \Comment{Update receiver's balance}
            \State $\received^{\id'} \gets \received^{\id'} :: C$
        \EndAsync
        \color{black}
    \EndCase
    \Case{$\ChangeKey(\pk')$}
        \State $\pk^\id \gets \pk'$ \Comment{Update authentication key}
    \EndCase
    \Case {$\StartConsensusInstance(\swid, \id_1, n_1, \id_2, n_2)$}
        \color{BlueViolet}
        \Async \Comment{Cross-shard request to $\swid$}
            \State run \Call{InitInstance}{$\swid, \id_1, n_1, \id_2, n_2, C$} \Comment{Create a new consensus instance}
        \EndAsync
        \color{black}
    \EndCase
\EndSwitch
\EndFunction
\end{algorithmic}
\end{algorithm}

\paragraph{Overview of the protocol.} The successive steps of an atomic swap of ownership between two accounts $\id_1$ and $\id_2$ can now be summarized as follows (see also Figure~\ref{fig:atomic_swap}):
\begin{itemize}
    \item The two owners of $\id_1$ and $\id_2$ coordinate off-chain and decide to swap the ownership between $\id_1$ and $\id_2$. After sharing the next sequence numbers $n_1$ and $n_2$ of their respective accounts ({\small\circled{1}}), they decide to ask a \emph{broker} to create a new UID for an atomic swap instance ({\small\circled{2}}). (This role may also be assumed by one of the owners.)
    
    \item The broker broadcasts an authenticated request containing an operation \[\StartConsensusInstance(\swid, \id_1, n_1, \id_2, n_2)\] for a suitable fresh UID $\swid$ ({\small\circled{3}}). After receiving a quorum of answers ({\small\circled{4}}), this results in a certificate $\Gamma$ certifying the creation of the instance $\swid$ to each client ({\small\circled{5}}).

    \item After verifying $\Gamma$, each owner broadcasts an authenticated request containing a $\LockInto$ operation in order to lock their account into $\swid$ ({\small\circled{6}}). This results in locking certificates $L_1$ and $L_2$ to be shared between clients ({\small\circled{8}}). 
    
    \item Based on $L_1$ and $L_2$, one of the clients (or both as long as they agree on the desired outcome) acting as a \emph{round leader} interact(s) with the consensus instance $\swid$ and attempts to drive the completion of the one-shot binary agreement protocol $\swid$ in order to confirm or abort the swap ({\small\circled{9}}).
    
    \item Eventually, at least one of the clients receives enough \emph{commit} votes from authorities running the consensus instance $\swid$ that it may create a \emph{commit certificate} $C^*$ (defined in the next paragraph).
    When a commit certificate $C^*$ is broadcast to authorities ({\small\circled{10}}), each authority issues internal cross-shard requests to the shards of $\id_1$ and $\id_2$, with the following effects:
    \begin{enumerate}
    \item the sequence number $\nextsequence^{\id_i}(\alpha)$ is incremented and $\pending^{\id_i}(\alpha)$ is reset to $\bot$ (effectively unlocking the account $\id_i$),
    \item $\confirmed^{\id_i}(\alpha)$ is updated to include $C^*$, and finally
    \item if the \emph{decision value} is $\Confirm$ and $C^*$ is seen for the first time, the authentication key of the account $\pk^{\id_i}(\alpha)$ is set to the appropriate key $\pk_{3 - i}$ initially chosen by the other owner as part of $L_{3 - i}$ ({\small\circled{11}}).
    \end{enumerate}
\end{itemize}

\begin{figure}[t]
    \centering
    \begin{tikzpicture}
        \node[draw,circle, minimum size=2.25cm] (0,0) {\textbf{Broker}};
        \node[draw,circle, minimum size=2.25cm] at (4,-4.3) {\textbf{Client 1}};
        \node[draw,circle, minimum size=2.25cm] at (10,-4.3) {\textbf{Client 2}};
        
        \tikzmath{\sidelen=3;}
        \tikzmath{\s=\sidelen/10;\inlen=\sidelen/3;}
        \begin{scope}[shift={(5,1.5)}]
            \draw[dashed] (0,0) rectangle (\sidelen,-1*\sidelen);
            \node[above] at (\sidelen/2,0) {\textbf{\sysname Committee}};
            
            \draw (\s,{-1*\s} ) rectangle ({\s+\inlen},{-1*(\s+\inlen)} );
            \draw ({\sidelen - \s - \inlen},{-1*\s} ) rectangle ({\sidelen-\s},{-1*(\s+\inlen)} );
            \draw (\s,{-1*(\sidelen - \s - \inlen)}) rectangle ({\s+\inlen},{-1*(\sidelen-\s)});
            \draw ({\sidelen - \s - \inlen},{-1*(\sidelen - \s - \inlen)}) rectangle ({\sidelen-\s},{-1*(\sidelen-\s)});
        \end{scope}

        \draw[solid,-latex, line width=0.25mm] (0.9,0.9) -- (4.9,0.9) node[midway,above]{\circled{3} request $\swid$}; 
        \draw[solid,-latex, line width=0.25mm] (4.9,0.2) -- (1.2,0.2)  node[midway,above]{\circled{4} votes for $\swid$};
        \draw[solid,-latex, line width=0.25mm] (1.2,0) -- (4.9,0)  node[midway,below]{\circled{5} cert for $\swid$};
       
        \draw[solid,-latex, color=BlueViolet, <->, line width=0.25mm] (5.0,-3.6) -- (9.0,-3.6) node[midway,above]{\bluecircled{8}\,locking certs $L_1, L_2$}; 

        \draw[solid,-latex, line width=0.25mm] (5.2,-4.5) -- (8.8,-4.5) node[midway,above]{\circled{5} cert for $\swid$};

        \draw[solid,-latex, color=BlueViolet, <->, line width=0.25mm] (5.0,-5.2) -- (9.0,-5.2) node[midway,below]{\bluecircled{1} $\id_1, n_1, \id_2, n_2$};
        
        \draw[solid,-latex, line width=0.25mm] (0.2,-1.2) -- (0.2,-4.0) -- (2.8,-4.0) node[midway,above, shift={(0.1, 0)}]{\circled{5}\,cert for $\swid$}; 

        \draw[solid,-latex, line width=0.25mm] (2.8,-4.6) -- (-0.2,-4.6) node[midway,below]{\circled{2} $\id_1, n_1, \id_2, n_2$} -- (-0.2,-1.2) ;

        \draw[solid,-latex, line width=0.25mm] (3.2,-3.3) -- (3.2,-0.8) node[midway,left]{\parbox{1.6cm}{\circled{6}\,locking request}} -- (4.9,-0.8) ; 

        \draw[solid,-latex, line width=0.25mm] (4.9,-1) -- (3.4,-1) -- (3.4,-3.2) node[midway,right]{\parbox{2cm}{\circled{7}\,locking votes}} ; 

        \draw[solid,-latex, line width=0.25mm] (9.8,-3.1) -- (9.8,-0.8) node[midway,right]{\parbox{2cm}{\circled{6}\,locking request}} -- (8.1,-0.8) ; 

        \draw[solid,-latex, line width=0.25mm] (8.1,-1) -- (9.6,-1) -- (9.6,-3.2) node[midway,left]{\parbox{1.7cm}{\circled{7}\,locking votes}} ; 

        \draw[solid,-latex, color=BlueViolet, <->, line width=0.25mm] (11.2,-4.) -- (12,-4.) -- (12,-0.25) node[midway,right]{\parbox{2.0cm}{\bluecircled{9} lead one consensus round}} -- (8.1,-0.25) ; 

        \draw[solid,-latex, line width=0.25mm] (11.2,-4.4) -- (14.8,-4.4) -- (14.8,0.3) -- (8.1,0.3) node[midway,above]{\circled{10} commit cert $C^*$} ; 

        \draw[thick, -latex, shorten >=1pt] (8.1, 0.8) to [out=10,in=+50,loop,looseness=6] (8.1, 1.3) node[right, shift={(0.65,0.35)}] {\circled{11} Swap owners and unlock $\id_1, \id_2$} ;
    
    \end{tikzpicture}
    
    \caption{An atomic swap}
    \label{fig:atomic_swap}
\end{figure}
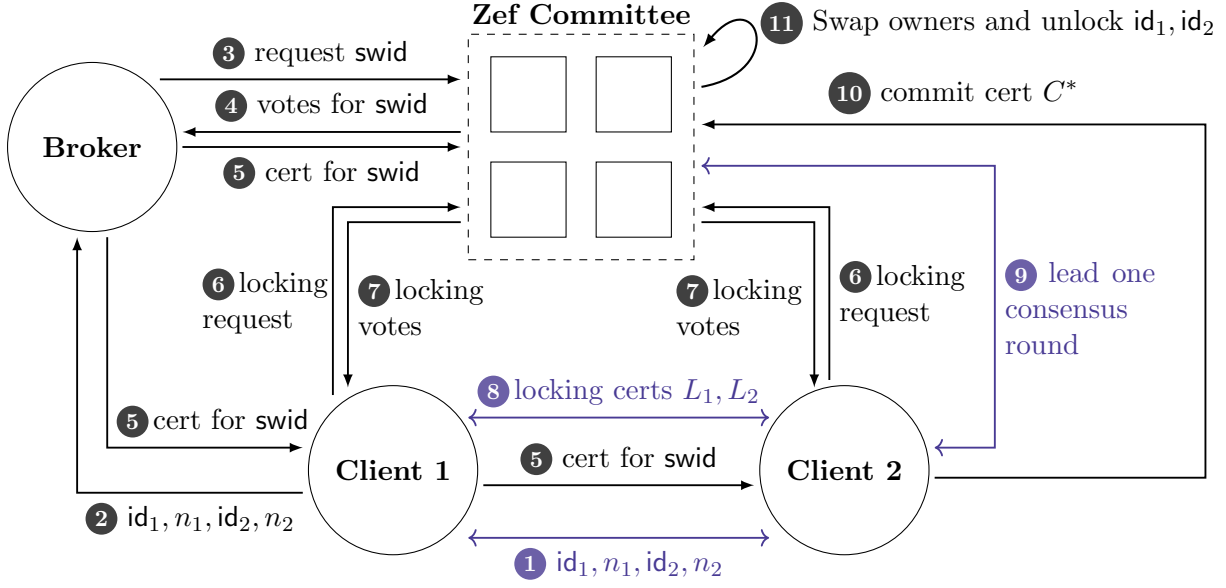

\paragraph{Data types and notations.} We introduce the following definitions:
\begin{itemize}
    \item A decision value $V$ is either $\Confirm$ or $\Abort$.
    \item A \emph{proposal} is a message $P = \Proposal(\swid, k, V)$ for some \emph{round number} $k \geq 0$ and decision value $V$. We write $\id(P) = \swid$, $\round(P) = k$ and $\decision(P) = V$.
    \item A \emph{pre-commit certificate} is a certificate on a proposal of the form $C = \cert[\PreCommit(P)]$.
    \item A \emph{commit certificate} is a certificate on a proposal of the form $C^* = \cert[\Commit(P)]$.
\end{itemize}
We extend the notations $\id(\cdot)$, $\decision(\cdot)$, and $\round(\cdot)$ to decision certificates and commit certificates.

By definition, \emph{agreement} holds iff two valid commit certificates for $\swid$ always contain the same decision value.

\paragraph{Atomic-swap states.} The state of an atomic-swap instance $\swid$ as seen by an authority~$\alpha$ can be described as follows:
\begin{itemize}
    \item The two accounts to swap $\id_i^\swid(\alpha)$ ($i \in \{1,2\}$);
    \item The two expected sequence numbers of lock certificates $n_i^\swid(\alpha)$;
    \item The two optional authentication keys $\pk_i^\swid(\alpha)$ (initially $\bot$ until the asset $i$ is \emph{locked});
    \item A last \emph{pending proposal} $\proposed^\swid(\alpha)$: initially $\bot$ then a proposal $P$.
    \item A last \emph{pending pre-commit}: $\locked^\swid(\alpha)$: initially $\bot$ then a pre-commit certificate~$C$.
\end{itemize}

\paragraph{Atomic-swap protocol.} An atomic-swap instance $\swid$ at $\alpha$ may receive the following requests $R$ from a client:
\begin{itemize}
    \item A proposal request $R = \HandleProposal(\auth_\pk[P], L_1, L_2)$ for some $P = \Proposal(\swid, k, V)$ and optional certificates $L_1$ and $L_2$.
    
    \item A pre-commit request $R = \HandlePreCommit(C)$ for some pre-commit certificate $C$.

    \item A commit request $R = \HandleCommit(C^*, L_1, L_2)$ where $C^*$ is a commit certificate and each $L_i$ is an optional lock certificate.
\end{itemize}

The handling of such requests to a consensus instance is presented in Algorithm~\ref{alg:consensus} and can be summarized as follows: (Invalid requests are ignored; additional \emph{safety} rules are provided below.)
\begin{itemize}
    \item Proposal request $R = \HandleProposal(\auth_\pk[P], L_1, L_2)$: The receiving authority $\alpha$ must verify that each $L_i$ is either $\bot$ or a \emph{valid lock certificate for the role} $i \in \{1, 2\}$ in the instance $\swid = \id(P)$, that is:
    \begin{itemize}
        \item $L = \cert[R]$ is a valid certificate for some $R = \Lock(\id, n, \LockInto(\swid, i, \pk))$;
        \item $\id = \id^\swid_i(\alpha)$ and $n = n^{\swid}_i(\alpha)$;
    \end{itemize}
    After setting $\pk^\swid_i(\alpha) = \pk$ for each such non-null $L_i$, the authority $\alpha$ verifies the following conditions:
    \begin{itemize}
        \item $P = \Proposal(\swid, k, V)$ is a proposal for $\swid$ and the authentication of $P$ by $\pk$ is correct;
        \item the proposal $P$ is \emph{valid} in the sense that $V = \Confirm$ implies that both input accounts are locked ($\forall i, \pk^\swid_i(\alpha) \neq \bot$);
        \item the round $k$ is \emph{available} (see discussion below);
        \item the proposal $P$ is \emph{safe} (see definition below)        
    \end{itemize}
    If the conditions are fulfilled then $\alpha$ sets $\proposed^\swid(\alpha)$ to $P$ and returns a signature on $\PreCommit(P)$.
    
    \item Pre-commit request $R = \HandlePreCommit(C)$: If $C = \cert[\PreCommit(P)]$ is a valid pre-commit certificate, $\id(C) = \swid$, and $C$ is safe (see below), then the authority sets $\locked^\swid(\alpha) = C$ and returns a signature on $\Commit(P)$.
    
    \item Commit request $R = \HandleCommit(C^*, L_1, L_2)$: If $C^* = \cert[\Commit(P)]$ is a valid commit certificate and $\id(C^*) = \swid$, then several cross-shard requests are prepared and sent to \emph{selected} accounts as follows:
    \begin{itemize}
    \item The account $\id$ is \emph{selected} iff either (i) it was locked previously in the instance ($\id = \id_i^\swid(\alpha)$ and $\pk^\swid_i(\alpha) \neq \bot$) or (ii) it holds that $\decision(P) = \Abort$ and a valid certificate $L_i = \cert[R_i]$ such that $R_i = \Lock(\id, n, \LockInto(\swid, i, \pk))$ is part of the request;

    \item Cross-shard requests are sent to the selected accounts $\id$ to unlock them by resetting $\pending^\id(\alpha)$ and incrementing $\nextsequence^\id(\alpha)$. If the decision value $\decision(P)$ is $\Confirm$ and the instance $\swid$ still exists, a new authentication key is also set for $\pk^\id(\alpha)$ thus fulfilling the desired swap of ownership.
    
    \item Finally, the instance $\swid$ is destroyed (if it was still present)
    \end{itemize}
\end{itemize}

Due to the validity condition above on $P$, the decision value is necessarily $\Abort$ if some account was never locked (i.e. $\pk^\swid_i(\alpha) = \bot$).
The lock certificates $L_i$ in the commit requests ensure that early termination of the consensus instance do not prevent users from unlocking their account in every authority afterwards.
To allow immediate deletion of an instance $\swid$, in the case of $\Abort$, we do not enforce consistency between the additional lock certificates $L_i$ and the original data in $\swid$. This bears no consequence since we only allow this behavior after verifying an $\Abort$ commit certificate for $\swid$.

\begin{algorithm}[t]
\caption{Consensus service}\label{alg:consensus}
\begin{algorithmic}[1]
\algdef{SE}[ASYNC]{Async}{EndAsync}{\textbf{do asynchronously}}{\algorithmicend}%
\algtext*{EndAsync}%
\algnewcommand\algorithmicswitch{\textbf{switch}}
\algnewcommand\algorithmiccase{\textbf{case}}
\algnewcommand\algorithmicassert{\texttt{assert}}
\algnewcommand\Assert[1]{\State \algorithmicassert(#1)}%
\algdef{SE}[SWITCH]{Switch}{EndSwitch}[1]{\algorithmicswitch\ #1\ \algorithmicdo}{\algorithmicend\ \algorithmicswitch}%
\algdef{SE}[CASE]{Case}{EndCase}[1]{\algorithmiccase\ #1:}{\algorithmicend\ \algorithmiccase}%
\algtext*{EndSwitch}%
\algtext*{EndCase}%

\small
\Function{InitInstance}{$\swid, \id_1, n_1, \id_2, n_2, C$}\Comment{Create a consensus instance for an atomic swap}
\For{$i \gets 1..2$}
    $(\id_i^\swid, n_i^\swid, \pk_i^\swid) \gets (\id_i, n_i, \bot)$
\EndFor
\State $(\proposed^\swid,\, \locked^\swid,\, \received^\swid) \gets (\bot,\, \bot,\, C)$
\EndFunction

\Statex

\Function{HandleProposal}{$\auth_\pk[P], L_1, L_2$} \Comment{Handle a proposal with optional lock certificates}
\State verify $\auth_\pk[P]$
\State match $\Proposal(\swid, k, V) = P$
\For{$i \gets 1..2$} \If{$L_i \neq \bot$}
    \State verify $\cert[R_i] = L_i$
    \State match $\Lock(\shorteq\id^\swid_i, \shorteq{n}^\swid_i, \LockInto(\shorteq\swid, \shorteq i, \pk_i)) = R_i$
    \State $\pk^\swid_i \gets \pk_i$ \Comment{Record the public key of role $i$}
\EndIf \EndFor
\State ensure $\pk \in \{ \pk^\swid_1, \pk^\swid_2 \}$ \Comment{Only users with a locked account can propose}
\State ensure $V = \Abort$ or $\forall i, \pk^\swid_i \neq \bot$ \Comment{Enforce validity of the swap}
\State ensure \Call{IsRoundAvailable}{\swid, k} \Comment{Available round values are restricted at a given time}
\State ensure \Call{IsSafeProposal}{P} \Comment{Enforce safety rule}
\State $\proposed^\swid \gets P$ \Comment{Record the proposal for future safety checks}
\State \textbf{return} $\Call{Vote}{\PreCommit(P)}$ \Comment{Success: return a signature meant to pre-commit $P$}
\EndFunction

\Statex

\Function{HandlePreCommit}{$C$} \Comment{Handle a pre-commit request}
\State verify $\cert[\PreCommit(P)] = C$
\State ensure \Call{IsSafePreCommit}{C} \Comment{Enforce safety rule}
\State $\locked^\swid \gets C$ \Comment{Record the pre-commit for future safety checks}
\State \textbf{return} $\Call{Vote}{\Commit(P)}$ \Comment{Success: return a signature meant to commit $P$}
\EndFunction

\Statex

\Function{HandleCommit}{$C^*, L_1, L_2$} \Comment{Handle a commit request}
\State verify $\cert[\Commit(P)] = C$
\State match $\Proposal(\swid, k, V) = P$
\For{$i \gets 1..2$}
    \State $(\id_i, n_i, \pk_i) \gets (\id^\swid_i, n^\swid_i, \pk^\swid_i)$ \Comment{Set locals with information from $\swid$ or $\bot$}
    \If{$V = \Abort$ and $L_i \neq \bot$} \Comment{Accept to unlock any account locked into $\swid$, once aborted}
        \State verify $\cert[R_i] = L_i$
        \State match $\Lock(\id, n, \LockInto(\shorteq\swid, \shorteq i, \pk)) = R_i$
        \State $(\id_i, n_i, \pk_i) \gets (\id, n, \pk)$
    \EndIf
\EndFor
\If{$V = \Confirm$ and $\swid \in \atomicswaps$}
    \For{$i \gets 1..2$}
        \color{BlueViolet}
        \Async \Comment{Cross-shard request to $\id_i$}
            \State ensure $\nextsequence^{\id_i} = n_i$
            \State $(\nextsequence^{\id_i}, \pending^{\id_i}) \gets (n_i + 1, \bot)$ \Comment{Unlock account $\id_i$}
            \State $\pk^{\id_i} \gets \pk_{3 - i}$ \Comment{Set the public key to the new value}
            \State $\confirmed^{\id_i} \gets \confirmed^{\id_i} :: C^*$
        \EndAsync
        \color{black}
    \EndFor
\Else
    \For{$i \gets 1..2$}
        \If{$\pk_i \neq \bot$}
            \color{BlueViolet}
            \Async \Comment{Cross-shard request to $\id_i$}
                \State ensure $\nextsequence^{\id_i} = n_i$
                \State $(\nextsequence^{\id_i}, \pending^{\id_i}) \gets (n_i + 1, \bot)$ \Comment{Unlock account $\id_i$}
                \State $\confirmed^{\id_i} \gets \confirmed^{\id_i} :: C^*$
            \EndAsync
            \color{black}
        \EndIf
    \EndFor
\EndIf
\State delete instance $\swid$ from $\atomicswaps$
\EndFunction
\end{algorithmic}
\end{algorithm}

\paragraph{Safety rules.} To guarantee agreement, an authority~$\alpha$ only accepts proposal and proposed certificates that are \emph{safe} at the time of the request (Algorithm~\ref{alg:safety_rules}):
\begin{enumerate}
    \item A proposal $P$ is safe for $\alpha$ iff the following conditions hold:
    \begin{itemize}
        \item (a) if $\bot \neq \proposed^\swid(\alpha) = P_0$ and $P \neq P_0$, then $\round(P) > \round(P_0)$;
        \item (b) if $\bot \neq \locked^\swid(\alpha) = C_0$, then $\round(P) > \round(C_0)$ and $\decision(P) = \decision(C_0)$.
    \end{itemize}
    \item A proposed certificate $C$ is safe for $\alpha$ iff the following conditions hold:
    \begin{itemize}
        \item (c) if $\bot \neq \proposed^\swid(\alpha) = P_0$, then $\round(C) \geq \round(P_0)$;
        \item (d) if $\bot \neq \locked^\swid(\alpha) = C_0$, then $\round(C) \geq \round(C_0)$.
    \end{itemize}
\end{enumerate}

Note that by definition of the protocol, $\proposed^\swid(\alpha)$ and $\locked^\swid(\alpha)$ never go back to~$\bot$ once there are set. Rather, these two fields respectively tracks the latest (safe) proposal~$P$ and the latest (safe) proposed certificate~$C$ that were voted on by $\alpha$.

\begin{algorithm}[t]
\caption{Safety rules}\label{alg:safety_rules}
\begin{algorithmic}[1]
\algdef{SE}[ASYNC]{Async}{EndAsync}{\textbf{do asynchronously}}{\algorithmicend}%
\algtext*{EndAsync}%
\algnewcommand\algorithmicswitch{\textbf{switch}}
\algnewcommand\algorithmiccase{\textbf{case}}
\algnewcommand\algorithmicassert{\texttt{assert}}
\algnewcommand\Assert[1]{\State \algorithmicassert(#1)}%
\algdef{SE}[SWITCH]{Switch}{EndSwitch}[1]{\algorithmicswitch\ #1\ \algorithmicdo}{\algorithmicend\ \algorithmicswitch}%
\algdef{SE}[CASE]{Case}{EndCase}[1]{\algorithmiccase\ #1:}{\algorithmicend\ \algorithmiccase}%
\algtext*{EndSwitch}%
\algtext*{EndCase}%

\small
\Function{IsSafeProposal}{$P$} \Comment{Determine if it is safe to vote for pre-committing $P$}
\State let $\Proposal(\swid, k, V) = P$
\If{$\proposed^\swid \neq \bot$ and $k \leq \round(\proposed^\swid)$}
\State return \textbf{false}
\EndIf
\If{$\locked^\swid \neq \bot$ and $\left(k \leq \round(\locked^\swid) \text{ or }V \neq \decision(\locked^\swid)\right)$}
\State return \textbf{false}
\EndIf
\State return \textbf{true}
\EndFunction

\Statex

\Function{IsSafePreCommit}{$C$} \Comment{Determine if it is safe to vote for committing $C$}
\State let $\cert[\PreCommit(P)] = C$
\State let $\Proposal(\swid, k, V) = P$
\If{$\proposed^\swid \neq \bot$ and $k < \round(\proposed^\swid)$}
\State return \textbf{false}
\EndIf
\If{$\locked^\swid \neq \bot$ and $k < \round(\locked^\swid)$}
\State return \textbf{false}
\EndIf
\State return \textbf{true}
\EndFunction

\end{algorithmic}
\end{algorithm}

\paragraph{Available rounds.} In practice, we may wish to prevent requests from using arbitrary round numbers $k \in \mathbb{N}$, because preventing exhaustion of such numbers would then require using unbounded-size infinite-precision integers. To address this issue while avoiding active coordination between authorities, we propose that authorities makes new round numbers available in sequential order at a fixed, given rate.

\paragraph{Client Protocol.} Assuming that an instance $\swid$ is still running and that no one else is proposing, a client with an asset locked in $\swid$ may drive completion as follows:
\begin{itemize}
    \item Query all the authorities in parallel to retrieve the highest round $k$ of a proposal $P=\proposed^\swid(\alpha)$ for some $\alpha$ and/or the highest pre-commit certificate $C = \locked^\swid(\alpha)$, if any.
    \item After a suitable delay $\delta$, if $C \neq \bot$, then broadcast $C$ to obtain a commit certificate. Otherwise, when $k+1$ is available, make a new proposal, then broadcast the pre-commit certificate.
    \item Broadcast the final commit certificate $C^*$.
\end{itemize}


\paragraph{Proof of safety.} The agreement property on commit certificates is derived from the following lemma:
\begin{lemma} Assume $C^* = \cert[\Commit(P_1)]$ and $C_2 = \cert[\PreCommit(P_2)]$ such that $\round(P_2) \geq \round(P_1)$ then $\decision(P_1) = \decision(P_2)$.
\end{lemma}
\begin{proof}
By induction on $\round(P_2) \geq \round(P_1)$.

If $\round(P_2) = \round(P_1)$, since the certificates $\cert[\PreCommit(P_1)]$ and $\cert[\PreCommit(P_2)]$ exist, by quorum intersection, there exists an honest node $\alpha$ that voted for both $P_1$ and $P_2$. However, by safety rule~(a), honest nodes only vote for new proposals with strictly increasing rounds, therefore $P_1 = P_2$.

Otherwise, assume $\round(P_2) > \round(P_1)$. Let $C_1 = \cert[\PreCommit(P_1)]$. By quorum intersection of $C^*$ and $C_2$, there exist be a honest node~$\alpha$ that voted for both $\Commit(P_1)$ and $\PreCommit(P_2)$.

By rule~(d), the round of $\locked^\swid(\alpha)$ never decreases. Thus, by rule~(c), $\round(P_2) > \round(P_1) = \round(C_1)$ implies that $\alpha$ voted for $\Commit(P_1)$ first, then $\PreCommit(P_2)$.

At the time of voting for $\PreCommit(P_2)$, by~(b) and~(d), we have that $\locked^\swid(\alpha) = C$ for some pre-commit certificate $C$ such that $\round(P_2) > \round(C) \geq \round(P_1)$ and $\decision(P_2) = \decision(C)$.
By induction, we conclude $\decision(P_1) = \decision(C)$.
\end{proof}

\paragraph{Discussion on termination.}
As noted earlier, after the two assets are locked, in theory, one of the owners can prevent the protocol from terminating by indefinitely submitting proposals that conflict with the other client. In addition to forfeiting half of the assets, this requires active communication with at least one honest authority at every new round in the future, when a round becomes available. Specifically, the malicious leader must indefinitely guess or quickly observe whether the other client is proposing $\Confirm$ or $\Abort$, and propose the opposite decision value.

A classical approach to enforce strict termination in the partially-synchronous model consists in (1)~restricting proposals to be signed by a particular client based on the parity of $k$ and (2)~making new rounds available at an exponentially slow rate. In practice, this approach may be activated after a certain delay, when it is clear that the two owners are not collaborating.

\paragraph{Comparison with existing consensus protocols.}
Our proposal is based on the observation that traditional leader-based consensus protocols do not technically require leaders to be drawn from the entire set of validators or even to have non-zero voting rights---as long as enough leaders can be trusted to make progress. In our case, this means that we can use the owner(s) of locked account(s) as leaders of the consensus protocol instead of Zef validators. Once both accounts are locked, our proposal lets the two leaders coordinate directly outside of the protocol --- at least for some time, until a slow, leader selection is (optionally) activated to enforce termination.

Compared to fully-featured implementation of a consensus protocol such as LibraBFT~\cite{libraBFT}, our approach is a one-shot consensus protocol, in particular chains of blocks are not needed. We also do not attempt to provide responsiveness or tight latency guarantees when leaders are not cooperative. However, our proposal is arguably significantly simpler, only uses constant storage, and does not require active coordination between authorities (e.g. broadcasting timeout messages).

\paragraph{Future work.} We have presented an atomic swap functionality that changes the owner's keys of the two accounts simultaneously. Arguably, this constitutes the first step towards a more general framework where multiple users may lock their accounts (and corresponding assets) into a consensus instance in order to execute arbitrary queries/updates on the locked accounts in an atomic way.

\begin{algorithm}[t]
\caption{Account service (unchanged from Zef)}\label{alg:account_service}
\small
\begin{algorithmic}[1]
\Function{InitAccount}{\id, \pk}\Comment{Initialize a new account}
\State $\pk^\id \gets \pk$
\State $\nextsequence^\id \gets 0$
\State $\balance^\id \gets \balance^\id(\initstate)$ \Comment{Initial balance is 0 except for special accounts}
\State $\confirmed^\id \gets [\,]$
\State $\received^\id \gets [\,]$
\EndFunction
\Statex
\Function{HandleRequest}{$\auth_\pk[R]$}\Comment{Handle an authenticated request from a client}
\State let $\Execute(\id, n, O) \;|\; \Lock(\id, n, O) = R$ \Comment{Allow regular and locking operations}
\State ensure $\pk^\id \neq \bot$ \Comment{Make sure the account is active}
\State verify that $\auth_\pk[R]$ is valid for $\pk = \pk^\id$ \Comment{Check request authentication}
\If {$\pending^\id \neq R$}
    \State ensure $\pending^\id = \bot$ and $\nextsequence^\id = n$ \Comment{Verify sequencing}
    \State ensure \Call{ValidateOperation}{$\id$, $n$, $O$} = $R$ \Comment{Validate the operation}
    \State $\pending^\id \gets R$ \Comment{Lock the account on $R$}
\EndIf
\State \textbf{return} $\Call{Vote}{R}$ \Comment{Success: return a signature of the request}
\EndFunction
\Statex
\Function{HandleConfirmation}{$C$}\Comment{Handle a certified request}
\State verify $\cert[R] = C$
\State match $\Execute(\id, n, O) = R$ \Comment{Allow regular operations only}
\State ensure $\pk^\id \neq \bot$ \Comment{Make sure the account is active}
\If {$\nextsequence^\id = n$}
    \State run \Call{ExecuteOperation}{$\id$, $O$, $C$}
    \State $\nextsequence^\id \gets n + 1$ \Comment{Update sequence number}
    \State $\pending^\id \gets \bot$ \Comment{Make the account available again}
    \State $\confirmed^\id \gets \confirmed^\id :: C$ \Comment{Append certificate to the log}
\EndIf
\EndFunction
\end{algorithmic}
\end{algorithm}

\section{Auctions}
\label{sec:auctions}

We now describe an extension of Zef~\cite{zef-report} meant to support running decentralised auctions. We target support both 1st price and 2nd price auctions:
\begin{itemize}
    \item In the 1st price auctions, a bidder with the highest bid gets the item and pays a price equivalent to its bid.
    \item In the 2nd price auctions, a bidder with the highest bid gets the item but pays a price equivalent of the 2nd highest bid in the auction.
\end{itemize}
The 2nd price auctions provide a truthfulness property[?], where all the bidders are incentivised to provide their true valuation of items (the winner never overpays for an item). However, it comes at a price of requiring additional security mechanisms. A malicious seller may participate in the auction (potentially with Sybil identities) uniquely to become the 2nd highest bid in the auction and thus increase the selling price of the item. Such a behavior can be disincentivised, if every submitted but unrevealed bid is penalized~\cite{ferreira2020credible}. 

In a classical setup, auctions require multiple phases that need to be ordered: 
\begin{itemize}
    \item the seller creates an auction
    \item the bidders need to submit their sealed bids
    \item the seller stops the bid submission
    \item the bidders need to reveal their bids
    \item the seller announces the result of the auction
\end{itemize}
A correct solution must order those operations to prevent race-conditions between the bidders and the seller (was a bid submitted/revealed before the deadline?).

\subsection{Bid submission}

We define a Threshold Public Key Encryption (TPKE) system consists of five algorithms~\cite{shoup2002securing}:
\begin{itemize}
    \item $\Setup(n, k, \Lambda)$: Takes as input the number of decryption servers $n$, a threshold $k$ where $1 \leq k \leq n$, and a security parameter $\Lambda \in \mathcal{Z}$. It outputs a triple $(\PK, \VK, \SK)$ where $\PK$ is a public key, $\VK$ is a verification key, and $\SK= (\SK_1, ..., \SK_n)$ is a vector of $n$ private key shares. Decryption server $i$ is given the private key share $(i, \SK_i)$ and uses it to derive a decryption share for a given ciphertext. The verification key $\VK$ is used to check validity of responses from decryption servers.
    \item $\Encrypt(\PK, m)$: Takes as input a public key $\PK$ and a message $m$ . It outputs a ciphertext $c$.
    \item $\ShareDecrypt(\PK, i, \SK_i, c)$: Takes as input the public key $\PK$, a ciphertext
    $c$, and one of the n private key shares in $\SK$. It outputs a decryption share $\mu= (i, \hat{\mu})$ of the enciphered message, or a special symbol $(i, \bot)$.
    \item $\ShareVerify(\PK, \VK, c, \mu)$: Takes as input $\PK$, the verification key $\VK$, a ciphertext $c$, and a decryption share $\mu$. It outputs valid or invalid. When the output is valid we say that $\mu$ is a valid decryption share of $C$.
    \item $\Combine(\PK, \VK, C, {\mu_1, . . . , \mu_k})$: Takes as input $\PK$, $\VK$, a ciphertext $c$, and $k$ decryption shares ${\mu_1, . . . , \mu_k}$. It outputs a cleartext $M$ or $\bot$.
\end{itemize}

We assume that the FastPay authorities jointly execute $\Setup(n, f+1, \Lambda)$ as a part of the bootstrap process. Each authority $a_i$ is given its private key share $\SK_i$, the public key $\PK$ and the verification key $\VK$. We assume that no other authority $a_j, j \neq i$ has access to $\SK_i$. \michal{We'll need to find a scheme providing that.}

A user $u_i$ willing to participate in the auction chooses an amount it is willing to pay for the item $v_i$. It also chooses $v_\keyword{max}, v_\keyword{max} \leq v_i$ to back its bid. The user transfers $v_\keyword{max}$ to the auction object. The money acts as a deposit that will be returned if the user does not win the auction or used to pay for the object if the user wins the auction. 

The user then encrypts its bid by executing $c_i = \Encrypt(\PK, m_i)$, where $m_i$ contains $v_i$ and the auction identifier. The user generates $z_i$, a proof of correctness of encryption $c_i$\michal{Need to figure out the details here}. 

The user submit $c_i$ and $z_i$ to the authorities and obtains a submission certificate $C_i$.\michal{Provide details on how the authorities verify correctness of the submission} Note that receiving a certificate (confirmation from $2f+1$ authorities), means that at least $f+1$ honest authorities have received the encryption and can jointly recover message $M_i$.

The bidders send their bid submission certificates offline to the seller. The seller then submits those certificates to the system. Effectively, the seller can choose which bids will be allowed to participate in the auction. However, the seller is incentivized to maximize the number of bids in the auction as each additional bid can only increase the selling price of the object (and thus the revenue of the seller). 

Once the seller decides that they gathered enough bids, they submits $\keyword{end of bidding}$ message and obtains a certificate on the submission. The $\keyword{end of bidding}$ message contains all the bids already submitted by the seller. After accepting an $\keyword{end of bidding}$ message the authorities stop accepting new bids. \michal{We can actually make it a single step here. Seems much more efficient}

The seller individually contacts each authority presenting a certificate on the $\keyword{end of bidding}$ message. If the certificate is valid, the contacted authority invokes $\ShareDecrypt$ on all the bids present in the message and releases its decryption shares $\mu_i$. The seller collects all the decryption shares and locally invokes $\Combine$ to recover the values of the bids. The seller includes all the decrypted bids in an $\keyword{end of auction}$ message and submits the message to the authorities (trying to get a certificate). 

Once a certificate on $\keyword{end of auction}$ is submitted to an authority, the authority:
\begin{itemize}
    \item Calculates the highest and the 2nd highest bid \michal{We might want to mention that we can further simplify this operation by incporating something like~\cite{krol2020pastrami}}
    \item Assigns the object being sold to the highest bidder
    \item Deducts the winner's deposit by the value of the 2nd highest bid and transfers this value to the seller
    \item Returns the deposits to the bidders
    \item Deletes the item objects
\end{itemize}

\michal{With the threshold encryption we know that the seller is able to decrypt all the bids. However, the seller might still not do the work: decrypt and submit all the bids. We know it's always his fault but need a way to penalize him. The easiest solution seems to be with a deposit. In ~\cite{ferreira2020credible}, the provide a formula to calculate the required amount. I wanted to avoid deposits altogether, but can't find a better solution here. A better approach would be to just penalize the bidders if they don't reveal, but then, they'll have to write to the auction object making things much more complicated.}

\michal{The one thing left is to deal with the problem of seller not advancing in the auction and the bidders not being able to recover their deposit. We can solve that using the one shot consensus algorithm.}

\printbibliography

\appendix

\section{Generalized Account Operations}
\label{sec:accounts}

In FastPay and Zef, each account state contains a balance, noted $\balance^\id(\alpha) \in \mathbb{Z}$. The execution of a payment operation can be seen as applying a pair of updates $(-x, +x)$ to the sender and the recipient states, respectively. Namely, $-x < 0$ is the \emph{local} update removing funds and $+x > 0$ is the \emph{remote} update adding funds. An authority accepts to validate an (authenticated) payment operation $(-x, +x)$ created by the owner of $\id$ iff the resulting local state $\balance^\id(\alpha) - x \geq 0$ is \emph{valid}. We note that remote updates $+x$ are always \emph{safe} in the sense that they never make a valid state invalid.

This leads us to propose the following general axioms for account states and updates:

\paragraph{Generalized states and updates.}
Let $\mathcal{S}$ be a set of \emph{state values} and $\mathcal{U}$ be a set of \emph{updates}. We assume a validity predicate $\mathsf{is\_valid}$ on $\mathcal{S}$, a safety predicate $\mathsf{is\_safe}$ on $\mathcal{U}$ and an operator $(\cdot)$ from $\mathcal{S} \times \mathcal{U}$ to $\mathcal{S}$ such that the following holds:
\begin{enumerate}
    \item $\forall s \in \mathcal{S}$, $\forall u_1, u_2 \in \mathcal{U}$, $s \cdot u_1 \cdot u_2 = s \cdot u_2 \cdot u_1$;
    \item $\forall s \in \mathcal{S}$, $\forall u \in \mathcal{U}$, $\mathsf{is\_valid}(s) \text{ and } \mathsf{is\_safe}(u) \Longrightarrow \mathsf{is\_valid}(s \cdot u)$.
\end{enumerate}

The commutative updates (1) as well as the notion of eventually consistency described in the proof of Zef (\cite{zef-report}, Section 4) draws similarities to the notion of Commutative Replicated Data Types (CmRDTs) in the field of distributed databases~\cite{crdts, study-crdts}.

\paragraph{Generalized account operations.}
We may generalize the protocol for direct payments of FastPay and Zef as follows:
\begin{itemize}
    \item The balance is replaced by a field 
    $\mathsf{state}^\id(\alpha) \in \mathcal{S}$ such that the initial value of a new account is always valid.
    \item 
    A new account operation $O = \mathsf{Apply}(\id', u_-, u_+)$ is  \emph{safe} to be issued by the owner of the account $\id$ as seen by an authority $\alpha$ iff $\mathsf{is\_safe}(u_+)$ and $\mathsf{is\_valid}(\mathsf{state}^\id(\alpha) \cdot u_-)$. (Note that in practice, additional constraints may apply on $u_-$ and $u_+$ for $O$ to be validated by $\alpha$.)
    \item The execution of an update $\mathsf{Apply}(\id', u_-, u_+)$ sent by account $\id$ consists in setting $\mathsf{state}^\id(\alpha) := \mathsf{state}^\id(\alpha) \cdot u_-$ and $\mathsf{state}^{\id'}(\alpha) := \mathsf{state}^{\id'}(\alpha) \cdot u_+$.
\end{itemize}

The same argument as in Section~4 of~\cite{zef-report} shows that whenever two honest authorities have executed the same certified updates then the two authorities agree on the states of active accounts. Besides, after every certified update has been executed by an authority $\alpha$, $\mathsf{is\_valid}(\mathsf{state}^\id(\alpha))$ holds for every $\id$.

This formalism lets us address the following applications:
\begin{itemize}
    \item A single NFT may be represented on-chain using $\mathcal{S} = \mathbb{Z}$ and $\mathcal{U} = \{-1, +1\}$. (Here, $\mathcal{S} = \mathbb{Z}$ ensures proper definitions of the operator $(\cdot)$. In practice, state values would range in $\{-1,0,1\}$, the invalid value $-1$ being a temporary state.)
    \item To support severals NFTs and several currencies at the same time, we note that independent updates on $\mathcal{S}_1 \times \mathcal{U}_1$ and $\mathcal{S}_2 \times \mathcal{U}_2$ may composed by defining a product operator $(\cdot)$ on $\mathcal{S} \times \mathcal{U}$ with $\mathcal{S} = \mathcal{S}_1 \times \mathcal{S}_2$ and $\mathcal{U} = \mathcal{U}_1 \biguplus \mathcal{U}_2$.
    \item A (multi)-set of objects or coins with monotonic requirements (e.g. to own X, one must be own parent(X) and 3 coins). (This approach is seen in CRDTs for data-structures such as trees, directed graphs, etc.)
\end{itemize}

\paragraph{Further generalization.} 
In the case of multi-currency states, we note that the formalism does not force $u_-$ and $u_+$ to be in the same unit of currency. Each validator may accept conversion requests up to a certain \emph{most favorable rate} that may differ from other validators and may fluctuate over time.

This approach paves the way for automated marker makers (AMM) in the Zef system: every shard of each Zef authority may maintain the current conversion rate(s) in a local cache and accept rate updates (``push") from a dedicated blockchain tracking the past transactions (using real-time aggregated counters in each authority) and continuously re-computing the rate(s).


\paragraph{Storage Cost}
A well-known issue of CmRDTs is the storage cost given that each update is individually synchronized and logged across replicas\footnote{\url{https://github.com/protocol/research-grants/blob/master/RFPs/rfp-005-optimized-CmRDT.md}}. In Zef, this issue is mitigated by the fact that accounts can be deleted after transferring their state to a new account, thereby effectively compressing the history of the account.

\end{document}